\documentclass{article}
\usepackage[affil-it]{authblk}

\usepackage[english]{babel}

\usepackage[letterpaper,top=2cm,bottom=2cm,left=3cm,right=3cm,marginparwidth=1.75cm]{geometry}

\usepackage{graphicx,amsmath,amsfonts,amsthm,amssymb}
\usepackage{graphicx}
\usepackage{booktabs}
\usepackage{tikz}
\usepackage{qtree}
\usepackage{tikz-qtree}
\usepackage[colorlinks=true, allcolors=blue]{hyperref}
\usepackage{xcolor}
\usepackage{caption,subcaption}
\usepackage{natbib}

\newtheorem{prop}{Proposition}

\newtheorem{definition}{Definition}
\newtheorem{theorem}{Theorem}

\newtheorem{problem}{Problem}

\title{Local Diversity of Condorcet Domains}
 
\author[1,2]{Alexander Karpov\footnote{Authors are  listed in alphabetical order}} 
\author[3]{Klas Markstr{\"o}m} 
\author[4]{S{\o}ren Riis}
\author[4]{Bei Zhou}
\affil[1]{HSE University, Moscow, Russia}
\affil[2]{Institute of Control Sciences, Russian Academy of Sciences, Moscow, Russia}
\affil[3]{Ume\aa\ University} 
\affil[4]{Queen Mary University of London}

\date{}

\begin{document}
\maketitle

\begin{abstract}
Several of the classical results in social choice theory demonstrate that in order for many voting systems to be well-behaved the set domain of individual preferences must satisfy some kind of restriction, such as being single-peaked on a political axis.  As a consequence it becomes interesting to measure how diverse the preferences in a well-behaved domain can be.

In this paper we introduce an egalitarian approach to measuring preference diversity, focusing on the abundance of distinct suborders one subsets of the alternative. We provide a common generalisation of the frequently used concepts of ampleness and copiousness. 

We give a detailed investigation of the abundance for Condorcet domains. Our theorems imply a ceiling for the local diversity in domains on large sets of alternatives, which show that in this measure Black's single-peaked domain is in fact optimal.   We also demonstrate that for some numbers of alternatives, there are Condorcet domains which have largest local diversity without having maximum order.

\end{abstract}

\section{Introduction}

In many situations studied in social choice a group of agents, which may be voters in a political system or artificial agents in a computer environment,  need to reach a common decision based on their individual preferences.  In situations where the individual preferences are completely unrestricted the now classical theorem of  \cite{arrow1950difficulty} shows that all decision procedures fail to satisfy several desirable properties.  Similarly the Gibbard-Satterwhaite theorem \cite{Gibbard,SATTERTHWAITE1975187} demonstrates the only social choice procedure which is proof against strategic manipulation is dictatorial.  On the other hand \cite{black}  showed that when preferences are based on preferences which are single-peaked on a common political axis, majority voting will be both decisive and strategyproof. These contrasting results lead to the question of how diverse a collection of individual preferences can be while also having a well-behaved social choice procedure. One of the most studied classes of well-behaved preferences is the Condorcet domain. A Condorcet domain is a set of linear orders such that the majority voting decision is transitive when each voter chooses their preferences from this domain. Condorcet domains secure strategyproofness in both deterministic \citep{Campbell} and randomized environments \citep{Brandt2023}. 

One of the most obvious measures for the diversity of opinions in a Condorcet domain is the size of the domain.  This is a well-studied  problem in social choice theory \cite{kim1992overview,puppe2023maximal}. Recent studies \citep{karpovslinko,setalternating} improved the asymptotic lower bound for the size of the largest Condorcet domains. Up to seven alternatives, all maximal Condorcet domains are investigated  in \cite{akello2023condorcet}, and in \cite{Leedham-Green2023condorcet} the largest domains on eight alternatives were determined. For higher numbers of alternatives the maximum possible size of a Condorcet domain remains an open problem.

However, the size is only the simplest measure of diversity and domains of the same size may have very different structures.  With this in mind, other types of diversity measures have been studied.  
\cite{puppe2018single} justified the single-peaked domain of preferences by requiring maximal diversity of the top elements. A more general measure of diversity is the support-based preference diversity index \cite{Hashemi2014} that evaluates diversity by the total number of distinct suborders of a given size. It reflects the idea of utilitarianism. It has substitutability, i.e. low diversity of preferences within one subset of alternatives can be compensated by high diversity within another subset of alternatives. By measuring the lowest number of distinct suborders over all subsets of a given size, we implement an egalitarian (Rawlsian) approach to preference diversity measurement. Note that this is egalitarian with respect to alternatives (subsets of alternatives), not voters/agents.

Recent literature has introduced two local diversity conditions for domains:  \emph{ampleness} and \emph{copiousness}, see \cite{slinko2019condorcet, puppe2023maximal}. The first means that the restriction of the domain to any pair of alternatives contains both of the possible orders, and the latter means that the restriction to any triple contains four orders. In this paper we generalize these notions by introducing a new diversity measure called abundance. Here we say that a Condorcet domain that contains at least $s$ suborders for each $k$-elements subset of alternatives is more diverse than a Condorcet domain that contains at least $s'$ suborders for each $k$-elements subset of alternatives if we have $s>s'$.

Some Condorcet domains have very low local diversity, e.g. non-ample Condorcet domains \cite{akello2023condorcet}. On the other hand, there are Condorcet domains with at least $2^{k-1}$ distinct suborders for each $k$-elements subset. Single-peaked, single-dipped, and group-separable domains are examples of such domains. Utilizing  Ramsey's theorem, we show that for sufficiently big numbers of alternatives $n$ and a fixed $k$, $2^{k-1}$ distinct suborders is the highest possible local diversity. This fact is an additional justification for the classical domains such as the single-peaked, and group-separable domains. It is possible to find significantly larger Condorcet domains than the single-peaked domain (see e.g. set-alternating scheme \cite{setalternating}), but their local diversity properties cannot be strictly better.

For a small number of alternatives, the Condorcet domain with the highest number of orders also has the best local diversity properties. Starting with $n=7$, there are maximal Condorcet domains that do not have the highest number of orders but have better local diversity parameters than the Condorcet domains of the maximal size. The study of local diversity properties thus opens a new line of research on large Condorcet domains.

In this paper we focus our discussion on Condorcet domains, however the concept of abundance applies to general domains and we will in our experimental analysis also include less restricted domains.

The structure of the paper is organised as follows: Section \ref{sec:background} introduces the notation and definitions employed throughout the paper. Section \ref{sec:local} defines the abundance of a domain and presents a series of theorems about abundance. Section \ref{sec:maximality} discusses the maximality of Condorcet domains, used in some of the later computational work. Section \ref{sec:ramsey} establishes asymptotic upper bounds for abundance via a Ramsey-type property for Condorcet domains. Section \ref{sec:abundance}  discusses abundance in terms of diversity indices and diversity orders for profiles. It follows this by an experimental study that compares the abundance of profiles randomly sampled from various domain types, including domains arising from empirical data. Finally, section \ref{sec:conclusion} gives some conclusions.

\section{Background and Notation}
\label{sec:background}
Let a finite set  $X=[n]=\{1, \dots, n\}$ be the set of alternatives. Let $L(X)$ be the set of all linear orders over $X$.  Each agent $i \in N$ has a preference order $P_i$ over $X$ (each preference order is a linear order). For brevity, we will write preference orders as strings, e.g. $12 \dots n$ means that $1$ is the best alternative, $n$ is the worst.

A subset of preference orders $D\subseteq L(X)$ is called a \emph{domain} of preference orders. A domain $D$ is a \emph{Condorcet domain} if whenever the preferences of all agents belong to the domain, the majority relation of any preference profile with an odd number of agents is transitive. A Condorcet domain $D$ is \emph{maximal} if every Condorcet domain $D'\supset D$ (on the same set of alternatives) coincides with $D$. A Condorcet domain $D$ is \emph{unitary} if it contains order $12\ldots n$. Each Condorcet domain can be rearranged to a unitary form by renaming alternatives. We will consider mainly unitary Condorcet domains.

Two Condorcet domains are isomorphic if they differ only by a relabeling of the alternatives.

The \emph{restriction} of a domain $D$ to a subset $A\subset X$ is the set of linear orders from $L(A)$ obtained by restricting each linear order from $D$ to $A$.

\cite{Sen1966} proved that  a domain is a Condorcet domain if the restriction of the domain to any triple of alternatives $(a,b,c)$ satisfies a never condition. A never condition can be of three forms $xNb$, $xNm$ $xNt$, referred to as a never bottom, a never middle, and a never top condition respectively. Here $x$ is an alternative from the   triple and  $xnb$, $xNm$, and $xNt$ means that $x$ is not ranked last, second, or  first respectively in the restricted domain. Fishburn noted that for domains with a societal axis never conditions can instead be described as  $iNj$, $i,j\in [3]$. $iNj$ means that $i^{th}$ alternative from the triple according to societal axis does not fall in $j^{th}$ place within this triple in any order from the domain. Unitary domains have at most six types of never conditions in Fishburn's format. 

A \emph{profile}  is a vector $\mathbf{R}=(R_1,R_2,\ldots,R_N)$ such that $R_i$ is the preferences of voter/agent $i$.  The \emph{support} $Supp(\mathbf{R})$ of a profile is the set of distinct preference orders occurring in $\mathbf{R}$. A profile is unanimous if the support has size 1. The \emph{census} of $\mathbf{R}$ is the vector $C(\mathbf{R})=((f_1,R_1),(f_2,R_2),\ldots)$ with one pair $(f_i,R_i)$ for each $R_i\in Supp(\mathbf{R})$ where $f_i$ is the number of agents with preference $R_i$.

From \cite{Hashemi2014} we have the following definitions.
A \emph{preference diversity order} is a partial order $\preceq $ on the set of profiles such that $\mathbf{R}_1\preceq \mathbf{R}_2$ if $\mathbf{R}_1$ is unanimous. 
A \emph{preference diversity index} is a function $\Delta$ from the set of profiles to the non-negative real numbers such that $\Delta(\mathbf{R})=0$ if $\mathbf{R}$ is unanimous. 

A diversity index $\Delta$ also defines a natural  diversity order $\preceq_\Delta$ by saying that $\mathbf{R}_1\preceq_\Delta  \mathbf{R}_2$ if $\Delta(\mathbf{R}_1)\leq \Delta(\mathbf{R}_2)$.

A Condorcet domain is said to be \emph{ample} if restriction of the domain to any pair of alternatives contains both of the possible orders, and \emph{copious} if the restriction to any triple contains four orders.

A domain which satisfies a never condition of the form $xN3$ for every triple is called an \emph{Arrow's single-peaked domain} \cite{arrow63}.
\emph{Black’s single-peaked domain} is the maximal Arrow’s single-peaked domain isomorphic to the domain defined by using the never condition  $2N3$ on all triples \cite{black}. 
A domain is a \emph{single-crossing domain} if there is an ordering for the set of preference orders in the domain such that for any pair of alternatives $(i, j)$ there exists a unique $k$ so that their relative order is $(i, j)$ for the top $k$ preference orders and $(j, i)$ for the rest \cite{Slinko2021}.   Fishburn's alternating domains are defined by assigning a triple $(i,j,k)$ the never condition $2N1$ if $j$ is odd and $2N3$ if $j$ is even \cite{fishburn1997acyclic}. \cite{setalternating} introduced the \emph{set-alternating} domains. Given a subset $A$ of the alternatives  the corresponding set-alternating  domain has the never condition $1N3$ on each triple $(i,j,k)$ where $j\in A$, and $3N1$ on the other triples.

Ramsey's theorem \cite{ramsey} asserts that given positive integers $k,r,n$ there is a positive integer $R(k, r, n)$ such that if $N\geq R(k,r,n)$, and each subset of size $r$ from $[N]$ is assigned a label from $[k]$, then all $r$-subsets from some subset of $[N]$ of size $n$ has the same label.

\subsection{Code and data}

For our computational results we have used the CDL library \cite{zhou2023cdl}, which provides functions for the study of Condorcet domains and other types of restricted preference orders.  This library includes a function which computes the restriction of a domain to a subset of the alternatives.

The programs and data used to derive the results in this paper are publicly available\footnote{GitHub. \url{https://github.com/sagebei/Local_Diversity_of_Condorcet_Domains.git}}.  

Whenever possible we have also verified our results by using independently coded functions in Mathematica.

\section{Local diversity conditions for Condorcet domains}
\label{sec:local}
There are two kinds of local conditions on domains that have been much used in recent literature, \emph{ampleness} and \emph{copiousness}. The first means that the restriction of the domain to any pair of alternatives will contain both of the possible orders, and the latter that the restriction to any triple will see the maximum possible four orders.

One way of looking at these conditions is that they require a minimum diversity of orders on sets of size 2 and 3 respectively. This point of view can be naturally generalised as follows.
\begin{definition}
    A domain $D$ is $(k,s)$-abundant  if the restriction of $D$ to any subset of $k$ alternatives has size at least $s$.

    A domain is \emph{exactly} $(k,s)$-abundant if it is $(k,s)$-abundant but not $(k,s+1)$-abundant. 
\end{definition}
Being $(3,4)$-abundant is equivalent to being copious and $(2,2)$-abundant is equivalent to being ample. Here we also note that being $(k,s)$-abundant is a monotone-increasing property under the addition of new orders to the domain; that is, for a given $k$  the domain $C\cup\{\sigma\}$ always has at least as large abundance as the domain $C$. This in turn implies that the largest possible abundance can always be found among the maximal Condorcet domains, but it can also occur already in non-maximal Condorcet domains. 

Let us look at the abundance properties of some classes of Condorcet domains.
\begin{theorem}
{\ }
    \begin{enumerate}
        \item The maximal Black's and Arrow's single-peaked domains on $n$ alternatives are exactly $(k,2^{k-1})$-abundant for every $k\leq n$. 
        
        \item Set-alternating domains \cite{setalternating} are $(k,2^{k-1})$-abundant for every $k\leq n$.   If $k<n/2$ then they are exactly $(k,2^{k-1})$-abundant. 

        \item Fishburn's alternating domains are exactly $(k,2^{k-1})$-abundant if $k\leq \lceil n/2\rceil+1$.
        
        For $n=4,5$ they are $(4,s)$-abundant for $s=9,8$.

        For $n=5,6,7$ they are $(5,s)$-abundant for $s=20,19,16$.
        
        For $n=6,7,8,9$ they are $(6,s)$-abundant for $s=45,42,39,32$.

        For $n=7,8,9,10,11$ they are $(7,s)$-abundant for $s=100,96$, $86, 79, 64$.
        
        \item The maximum Condorcet domain for $n=8$ is $(4,8),(5,18),(6,40),$ and $(7,96)$-abundant.
    \end{enumerate}
\end{theorem}
\begin{proof}
        Point 1 follows since the restriction of an Arrow's single-peaked domain to a subset of the alternatives is also an Arrow's single-peaked domain, and by \cite{slinko2019condorcet}  any maximal Arrow's single-peaked domain on $k$ alternatives has size $2^{k-1}$.

         Point 2  follows using the construction in \cite{setalternating} of set-alternating domains. It is shown that every such domain on $t$ alternatives has  size at least $2^{t-1}$, and the restriction of a set-alternating domain to a subset of the alternatives is also a set-alternating domain.

         The first part of point 3 follows by taking the restriction of the domain to the set of even integers together with 1 $n$, if the latter is odd.  This set has size $\lceil n/2\rceil+1$  and the restricted domain uses only the never condition $2N1$.  If $k$ is smaller than $\lceil n/2\rceil+1$ we further restrict to an arbitrary subset $A$ of size $k$.  By \cite{raynaud1981paradoxical} the restricted domain has size $2^{k-1}$, since it uses a single never condition. 

        The remainder of point 3 follows from a computational test of these domains, done independently in Mathematica and CDL.
         
         Point 4 follows by a computational test of the maximum domain from \cite{Leedham-Green2023condorcet}. 
\end{proof}

\begin{theorem}
    For $n\leq 6$ and $k=2,3,4,5,6$, Fishburn's alternating domains have the maximum possible abundance among all maximal Condorcet domains with a given number of alternatives.  For $n=7$, they do not reach the maximum possible for $k=6$.
\end{theorem}
\begin{proof}
    The abundances for all maximal Condorcet domains in this range were computed using the data from \cite{akello2023condorcet}.
\end{proof}
For  $n=7$, $n=8$, with $k=6$, there are maximal Condorcet domains with higher abundance than the Condorcet domains of maximum size.

Given two domains with large abundances, we have a natural construction for larger domains which inherit some abundances from the first domains.
\begin{definition}
    Given two domains $D_1$ and $D_2$ on the sets $A$ and $B$ respectively we define a domain $S(D_1,D2)$ on the set  $A\cup B$ as follows:  Given any order $u$ from $D_1$ and any order $v$ from $D_2$ the domain contains both the concatenation $uv$ and the concatenation $vu$.
\end{definition}
The domain $S(D_1,D2)$  has size $2|D_1||D_2|$ and cannot have low abundance when $D_1$ and $D_2$ do not.
\begin{prop}
    If $D_1$ and $D_2$ are $(k,2^{k-1})$-abundant for each $k\leq t$ then so $S(D_1,D_2)$. 
\end{prop}
\begin{proof}
    This obviously holds for $k$-tuples which only contain alternatives from one of $A$ and $B$.  Let us assume that the tuple has elements $T_1$ from $A$ and $T_2$ from $B$. Then $S(D_1, D_2)$ contains both every concatenation of the, at least,  $2^{|T_1|-1}$ elements in the restriction of $D_1$ to $T_1$ with the, at least, $2^{|T_2|-1}$ elements in the restriction of $D_2$ to $T_2$,  and the same number of concatenations om the opposite order.

    This gives a total of at least $2\times2^{|T_1|-1}\times2^{|T_2|-1}=2^{k-1}$ linear orders in the restriction to this $k$-tuple.
\end{proof}

Given that abundance is a monotone-increasing property it is natural to ask how many linear orders we must have in order to achieve a specified abundance. This is equivalent to asking,  given a fixed $k$ how small can a \emph{minimal} $(k,s)$-abundant domain $D$ be?  Here minimal means that every subset of $D$ is a domain which is not $(k,s)$-abundant.
\begin{theorem}
We have the following existence results for small domains with given abundance: 
    \begin{enumerate}
        \item For every $k\geq 2$ there are $(k,2)$-abundant domains of size 2 for every $n\geq 2$.
        \item There exists $(3,3)$-abundant and $(3,4)$-abundant domains of size 3 and 4 respectively if and only if $3\leq n\leq 8$.  
        \item There exists single-crossing maximal Condorcet domains with size ${n \choose 2}+1$ which are $(3,4)$-abundant 
    \end{enumerate}
\end{theorem}
\begin{proof}
    \begin{enumerate}
        \item Take the domain which consists of the natural order and its reverse.
        \item One such domain for $n=8$ is displayed in Figure \ref{fig:abundant} (rows represent linear orders). The first three orders form a $(3,3)$-abundant domain, and by deleting alternatives we reach domains for smaller $n$ with the same abundance. A computer search shows that no $(3,3)$-abundant domains exist for $n=9$.
        \item See \cite{Slinko2021} for an example of a maximal single-crossing domain that is copious and has size ${n \choose 2}+1$.
    \end{enumerate}
\end{proof}
\begin{figure}[h]
\centering
\begin{tabular}{cccccccc}
\toprule
1&2&3&4&5&6&7&8\\
\hline
4&3&2&1&8&7&6&5\\
\hline
6&5&8&7&2&1&4&3\\
\hline
7&8&5&6&3&4&1&2\\
\bottomrule
\end{tabular}
\caption{A $(3,4)$-abundant domain of size 4 for $n=8$}\label{fig:abundant}
\end{figure}

As the example in Figure \ref{fig:abundant} (each row is an order) shows it is possible to have a large abundance with a, perhaps, surprisingly small number of voters.  However, examples of the type given here are rare, for $n=8$ there is up to isomorphism just 7 such Condorcet domains. As a comparison, for a set of $t$ randomly chosen orders from the unrestricted domain with $n=8$ reaches expected abundance $(3,4)$ at $t=14$. It would be interesting to investigate how large the abundance in empirical domains typically is.

\subsection{Diversity inheritance}
Being abundant  for one set of parameters often implies that a domain is abundant for additional parameters. Using data from \cite{akello2023condorcet} one can  establish the  inheritance properties presented in Table \ref{tab:inhet}. Here  $(k_1, s_1) \rightarrow{} (k_2, s_2)$  means that every domain which is $(k_1, s_1)$-abundant is also  $(k_2, s_2)$-abundant.

\begin{table}[ht]
\centering
\caption{Implications for induced abundance}\label{tab:inhet}
\label{tab:induceabound}
\begin{tabular}{lll}
\toprule
 $(3, 4) \rightarrow{} (2, 2)$ & $(4, 9) \rightarrow{} (2, 2)$ & \\
 $(5, 19) \rightarrow{} (2, 2)$& $(6, 43) \rightarrow{} (2, 2)$ & \\
\midrule
$(4, 4) \rightarrow{}  (3, 1)$ & $(4, 5-6) \rightarrow{}  (3, 2)$ & $(4, 7-8) \rightarrow{}  (3, 3)$ \\
$(4, 9) \rightarrow{}  (3, 4)$ & &\\
\midrule
$(5, 12) \rightarrow{} (3, 1)$ & $(5, 13-16) \rightarrow{} (3, 2)$ & $(5, 17-18) \rightarrow{} (3, 3)$  \\
$(5, 19-20) \rightarrow{} (3, 4)$ & & \\
\midrule
$(6, 40) \rightarrow{} (3, 2)$ & $(6, 41-42) \rightarrow{} (3, 3)$ & $(6, 43) \rightarrow{} (3, 4)$ \\
\midrule
$(5, 5) \rightarrow{} (4, 1)$ & $(5, 6-8) \rightarrow{} (4, 2)$  & $(5, 9-11) \rightarrow{} (4, 3)$   \\
$(5, 12-13) \rightarrow{} (4, 4)$ & $(5, 14-15) \rightarrow{} (4, 5)$  & $(5, 16-17) \rightarrow{} (4, 6)$  \\
$(5, 18) \rightarrow{} (4, 7)$ & $(5, 19-20) \rightarrow{} (4, 8)$ &  \\
\midrule
$(6, 40-41) \rightarrow{} (4, 6)$&$(6, 42-43) \rightarrow{} (4, 7)$  &$(6, 44-45) \rightarrow{} (4, 8)$  \\
\midrule
$(6, 37) \rightarrow{} (5,13 )$ & $(6, 38) \rightarrow{} (5, 13)$& $(6, 39) \rightarrow{} (5, 14)$\\
$(6, 40) \rightarrow{} (5, 15)$ & $(6, 41) \rightarrow{} (5, 16)$ &$(6, 42) \rightarrow{} (5, 16)$ \\
$(6, 43) \rightarrow{} (5, 17)$ & $(6, 44) \rightarrow{} (5, 18)$    & $(6, 45) \rightarrow{} (5, 19)$ \\
\bottomrule
\end{tabular}
\end{table}
Here we point out that it is only for $k=3$ that we see an implication of the form $(k,2^{k-1}) \rightarrow{} (k-1,2^{k-2})$, corresponding to the known fact that being copious implies that the domain is also ample. This demonstrates that for $k\geq 4$ the abundance properties are genuinely distinct. As the table shows, the inherited abundance drops quickly when $k$ is increased.

\section{A note on maximality}
\label{sec:maximality}
In this section we will consider the restrictions of a maximal Condorcet domain to all subsets of alternatives of size $n-1$.  This is done for two reasons. First, this will give us information about how compatible the rankings of alternatives from different subsets are, and as we  will see there are domains where this is so context dependent that none of these restrictions are maximal domains. Second, domains which can be restricted to maximal subdomains are sometimes easier to treat computationally and some of the results in the next section will depend on this.

When considering the collection of restrictions to $n-1$ alternatives  there are in principle three cases. Either all restrictions to $n-1$ alternatives are maximal Condorcet domains, or there are both maximal and non-maximal restrictions, or all such restrictions are non-maximal Condorcet  domains. In the last case we have mutually incompatible preferences coming from all subsets of the alternatives.

Let us say that a Condorcet domain $D$ on $n$ alternatives is \emph{discordant} if $D$ is maximal  but  does not have a subdomain $D'$ on $n-1$ alternatives which is also a maximal Condorcet domain. By eliminating an alternative from a discordant domain we resolve conflicting never conditions, and we can find some orders that do not belong to the restriction, but satisfy all never conditions.

Note that the copiousness property of \cite{slinko2019condorcet} is equivalent to saying that every subdomain on $3$ alternatives  is maximal. Slinko gave examples of domains which are not copious, and many such examples were given in \cite{akello2023condorcet}. However, all the non-copious examples for small $n$ also have subdomains on 3 alternatives which are maximal, i.e. they have triples of both kinds.  In fact, using the data from \cite{akello2023condorcet} exhaustive testing shows that for $4\leq n\leq 5$ there are no discordant domains.  However, by extending this exhaustive search to the domains for $n=6, 7$ are found.

We find that for $n=6$ there are in total five discordant domains. One of these has size 4 and the other four have size 25. We display the example of size 4 domain in Figure \ref{fig:odd}. The restriction of this domain to each five-alternatives subset contains four orders, but the corresponding maximal Condorcet domain on five alternatives has eight orders.

\begin{figure}[h]
\centering
\begin{tabular}{cccccccc}
\toprule
1&2&3&4&5&6\\
\hline
3&6&2&5&1&4\\
\hline
4&1&5&2&6&3\\
\hline
6&5&4&3&2&1\\
\bottomrule
\end{tabular}
\caption{The smallest discordant Condorcet domain }\label{fig:odd}
\end{figure}

For $n=7$ we have a much wider range of domain sizes for which discordant domains exist.  In Table \ref{tab:disco} we show the number of non-isomorphic discordant Condorcet domains of each size for $n=7$. These domains can be downloaded from \cite{Web1}.  

\begin{table}[h]
\caption{The number of discordant domains and their sizes for $n=7$}\label{tab:disco}
\centering
\begin{tabular}{l|*{10}{c}}
\toprule
Size &12 & 14 & 15 & 16 & 17 & 18 & 19 & 20 & 21 & 22  \\
Count &2 & 4 & 2 & 20 & 14 & 8 & 66 & 48 & 48 & 24  \\
\hline
Size & 23 & 24 & 25 & 26 & 27 & 28 & 29 & 30 & 31 & 32 \\
Count & 8 & 28 & 32 & 38 & 26 & 38 & 32 & 69 & 36 & 89 \\
\hline
Size & 33 & 34 & 35 & 36 & 37 & 38 & 39 & 40 & 41 & 42 \\
Count     & 64 & 30 & 90 & 102& 78 & 70 & 70 & 304& 34 & 49 \\
\hline
Size & 43 & 44 & 45 & 46 & 47 & 48 & 49 & 50 & 51 & 52  \\
Count  & 44 & 95 &92 & 74 & 64 & 94 & 56 & 110 & 42 & 100  \\
\hline
Size& 53 &54  & 55 & 56 & 57 & 58 & 60 &&&\\
Count&38 & 78 & 10 & 74 & 4 & 48 & 15 &&&\\
\bottomrule
\end{tabular}

\end{table}

As we can see, discordant  domains exist for most small sizes for $n=7$ but do not exist for any size above 60.  For domains close to the maximum possible size it seems likely that restrictions will also remain large and we pose the following problem:
\begin{problem}\label{maxprob}
    Let $f(n)$ denote the maximum possible size for a Condorcet domain on $n$ alternatives.
    
    Is it true that every discordant Condorcet domain has size at most $0.8 f(n)$?
\end{problem}
The constant 0.8 in the problem is not intended to be sharp, but rather is a description of the intuitive picture that very large domains should be so diverse that this is true for their restrictions as well. From our data, we see that for $n\leq 7$ the problem would have a positive answer for a constant much smaller than 0.8.

\section{Ramsey properties of diversity}
\label{sec:ramsey}

From  e.g. Black's single-peaked domain we know that  $(k,2^{k-1})$  domains exist for all $n\geq k$. However, for a fixed $k\geq 4$ there are values of $s$ such that $(k,s)$-abundant domains exist only for a finite set of values of $n$. This connects to the observation in \cite{zhou2023new} that large enough domains should have subdomains of size six using only a single never condition.  We here provide a proof using the hypergraph version of Ramsey's theorem and also note that this implies the existence of a more specifically order-theoretic Ramsey property.

\begin{theorem}
    Given an integer $k$ there exists two numbers $R_s(k) \leq R_u(k)$ such that 
    \begin{itemize}
        \item For $n\geq R_s(k)$, no domain on $n$ alternatives is $(k,s)$-abundant for $s>2^{k-1}$.
        \item For $n\geq R_u(k)$, every domain on $n$ alternatives has a subset $A'$ of $k$ alternatives such that every triple in $A'$ satisfies the same never condition.
    \end{itemize}
\end{theorem}
\begin{proof}
    Given a unitary Condorcet domain $C$ on $n$ alternatives we define an edge colouring of the complete 3-uniform hypergraph $K_n^{(3)}$ by using the name of the never condition on a triple as the colour of the corresponding edge of the hypergraph. This colouring uses at most 6 colours.

    By Ramsey's theorem for hypergraphs \cite{ramsey} there exists an integer $N$ such that if $n\geq N$ then there must exist a subset $A'$ of the vertices such that every triple has the same colour. This proves the existence of  $R_u(k)$.

    By Theorems 3 and 4 of \cite{raynaud1981paradoxical} a domain on $k$ alternatives in which every triple satisfies the same never condition has size $2^{k-1}$. This proves that $R_s(k)$ exists and is at most $R_u(k)$.
\end{proof}

The general theorem shows that $R_s(k) \leq R_u(k)$  but for small values of $k$ we find that $R_s(k)$ is typically much smaller than $R_u(k)$, which is equal to the corresponding Ramsey number. The latter is known to grow super-exponentially in $k$ \cite{MR2552253} and the only known exact value is that with two colours an $k=4$ we get $R=13$ \cite{MR1095846}.

We have computed the exact value of $R_s(k)$  for several small values of $k$.
\begin{theorem}
    The following bounds on the number of alternatives in a $(k,s)$-abundant domain hold:
    \begin{enumerate}
        \item $(4,9)$-abundant domains exist only for $n=4$.
        \item $(5,20)$ and $(5,19)$-abundant domains exist if and only $5\leq n\leq 6 $.
        \item $(5,18)$ and  $(5,17)$-abundant domains exist if and only if $5\leq n\leq 8$.
        \item $(6,43)$, $(6,44)$, and    $(6,45)$-abundant domains exist only for $n=6$.
        \item $(6,42)$-abundant domains exist  for $n=6,7$. If problem \ref{maxprob} has an affirmative answer, they do not exist for $n\geq 8$.
        \item $(6,41)$-abundant domains exist  for $n=6,7,8$. If problem \ref{maxprob} has an affirmative answer, they do not exist for $n\geq 9$.
    \end{enumerate}
\end{theorem}
\begin{proof}
    Points 1 to 4 follow by exhaustively checking the abundance for the maximal CDs generated in \cite{akello2023condorcet} and the maximum CD from \cite{Leedham-Green2023condorcet}, followed by a search for $n=9$ using CDL. 

    For point 5, the existence of $(6,42)$-abundant domains again follows by analysing the data set from \cite{akello2023condorcet}. From that data we found all maximal  $(6,42)$-abundant domains for $n=7$,  all these domains had size 92 or more. Using CDL we found all extensions of these domains to $n=8$ and checked that none were  $(6,42)$-abundant. 

    For point 6,  the existence of $(6,41)$-abundant domains for $n=6,7$ again follows by analysing the data set from \cite{akello2023condorcet}. From that data, we found all maximal  $(6,41)$-abundant domains for $n=7$,  all these domains had size 91  or more. Using CDL we found all extensions of these domains to $n=8$, where examples of such domains were found. These were in turn extended to $n=9$  checked that none were  $(6,41)$-abundant. 
\end{proof}
Abundance $(5,18)$ is achieved for $n=8$ by the maximum domain from \cite{akello2023condorcet},  and for $n=7$ an example of size 96 can be constructed as a subdomain of the $n=8$ domain. 

Abundance $(6,42)$ for $n=7$ is achieved by Fishburn's alternating domain.

Abundance $(6,41)$ is achieved for $n=8$ by two maximal domains of size 219.  The first of these can be viewed as a variation of Fishburn's alternating domain where the never conditions on four triples have been changed as follows: ((1, 2, 4),1N3),((1, 2, 6),1N3), ((3, 7, 8),3N1),((5, 7, 8),3N1).


\section{Abundance based diversity indices}
\label{sec:abundance}

Our next aim is to connect the local diversity in the form of abundance to the study of diversity indices, as outlined by Hashemi and Endriss \cite{Hashemi2014}. 

We define the \emph{abundance vector}  $\mathcal{A}(D)$ of a domain $D$ to the vector in which entry number $k$ is the $s$ such that $D$ is exactly $(k,s)$-abundant.  Here position $k=1$ will be $1$ for every non-empty domain.

We say that $D_1$ is ranked higher than $D_2$ if $\mathcal{A}(D_1)$  is lexicographically larger than $\mathcal{A}(D_2)$ 

This diversity order emphasises diversity according to set size, so that the smallest sets are the most significant.  This is often the prioritisation since two domains can easily have the same, large, size without having the same abundance even for $k=2$.

Hashemi and Endriss \cite{Hashemi2014} considered, among others, two basic diversity indices. The \emph{simple support-based index} $\Delta_{supp}(D)=|D|-1$  and the \emph{support-based index} $\Delta_{supp}^k=|\cup_{A\subset X, |A|=k}Supp(D(A))|-{n \choose k}$.
Given a fixed integer $k$ we can similarly define a simple diversity index  by taking the  largest $s$  such that $Supp(\mathbf{R})=D$ is $(k,s)$-abundant. 
This index will rank profiles, and domains,  specifically according to their diversity on the restriction to $k$-sets. Here $\Delta_k(D)= s$ if $D$ is exactly $(k,s)$-abundant.

Recall that the binary entropy function is given by $$H(x)=-x\ln(x)-(1-x)\ln(1-x)$$
If two profiles have the same abundance index $\Delta_k(\mathbf{R}_1)=\Delta_k(\mathbf{R}_2)$ we can instead use the entropy version of the abundance index $\Delta_k(D)^e=\min_{A\subset X, |A|=k} H(D(A))$, where $H(D(A))$ is the entropy of the census counts $f_i/N$. Note that $H(D(A))$ is maximized if the profile is uniform over the preferences in $A$ and minimized if as many agents as possible use  a single order from $D(A)$. This index can either be seen as measuring how uniform the agent's preferences are on $A$ or as a proxy for how robustly $(k,s)$-abundant the profile is. Here the robustness refers to how many preference orders would one have to change in order to reduce the abundance for $k$-subsets.

The highest entropy is seen for maximal group-separable domains. For each pair of alternatives each group-separable domain is partitioned on two equal-sized parts with different preferences over the pair of alternatives. For each subset of alternatives of size $k$, each group-separable domain is partitioned on $2^{k-1}$ equal-sized parts with different preferences over the subset alternatives.

A less egalitarian alternative for comparing domains with equal abundance vector is to use the sum of the abundances for all $k$-subsets: $\sum_{|A|=k} |D(A)|$  
or, for a profile,  the sum of the entropies $\sum_{|A|=k} |H(D(A))|$.

\subsection{Experimental analysis}
In order to demonstrate the concepts developed in this paper we applied them to both profiles based on empirical data and profiles based on sampling from different types of known domains, both Condorcet and not.

Our empirical  data comes from the AGH data set\footnote{PrefLib. \url{https://www.preflib.org/dataset/00009}}  from \cite{MaWa13a}. This is a data set based on the choices of courses for a group of students during two consecutive years. There were 9  available courses in 2003 and 7 in 2004. The first year had 123 students and the second 70. Neither year leads to a CD. This data set was also used for simulations 
in \cite{Hashemi2014}.

We first investigate the empirical data as it is. The support is not a Condorcet domain for either of the two years. Using the full data set and all alternatives the two years give abundance vectors 
$$(1,1,2,6,16,39,75,101),$$
$$(1,1,2,6,16,39).$$
The domains are non-ample due to a single course which every student takes.  If  we delete that course both years lead to an ample domain.
If we delete the mandatory course and restrict to the 5 most popular courses we get the vectors
$$(1,2,5,15,31),$$
$$(1,2,6,16,32).$$
Here we see that despite the first profile having more students, 123 vs 70, it is the second one which has higher abundance.

We have also made a study for different kinds of randomly sampled profiles, similar to that in  \cite{Hashemi2014}.  We generate random samples from six different models, and in each case we generate domains with $n=7$ alternatives and $N=50$ agents. In the first model we pick preferences uniformly at random from Black's single-peaked domain for $n=5$. In the second model we pick orders from the 2004 AGH data restricted by removing alternatives 8 and 9. In the third model we pick orders from the unrestricted domain for $n=7$, a probabilistic model known at the Impartial culture. In the next models we pick preferences from Fishburn's alternating domain, a group-separable/symmetric domain, and a (5,18)-abundant domain of size 96. 

In Figure \ref{fig1} we show the empirical distribution for the abundance when 2000 domains are sampled uniformly from each model.  For both $k=3$ and $k=4$ we see that the unrestricted domain leads to the largest expected abundance, and the single-peaked domain with given axis to the lowest. Among Condorcet domains the group-separable domains (caterpillar group-separable, see \cite{Faliszewski2022}) give the highest expected abundance, except for $k=6$. In this domain different suborders emerge with  equal probabilities, and we have high abundance with a higher probability, than in domains in which diversity is based on several rare orders. For $k=6$, (5,18)-abundant domain surpasses the group-separable domain. It confirms the superiority of the investigated domain.

The data-based model has higher diversity than Condorcet domains, and since the support is not a Condorcet domain one might expect more varied preferences there.

Most of the  abundance  distributions are unimodal or bell-type for the chosen parameters. However, for each size $k$ we will see the abundance concentrating on the maximum possible value for the underlying domain  when the number of agents is large enough. In particular the classical Coupon collector problem shows that if the support is $(k,s)$-abundant then a uniformly sampled profile will typically also be that if the number of agents substantially exceeds $s \log{s}$.  However, for a well-structured support one might see such concentration already for smaller numbers of agents.

The abundance (diversity) axis is an additional dimension which can be used to distinguish statistical cultures (see the survey \cite{AAMAS20} about statistical cultures). Other such distinguishing indices have been studied  but unlike  e.g.  the   distance-based indices utilized in \cite{Faliszewski_diversity}  abundance calculation is quite simple, as long as $k$ is moderate in size.

\begin{figure}[ht]
\centering
	\begin{subfigure}{0.4\linewidth}
	 \includegraphics[width=\textwidth]{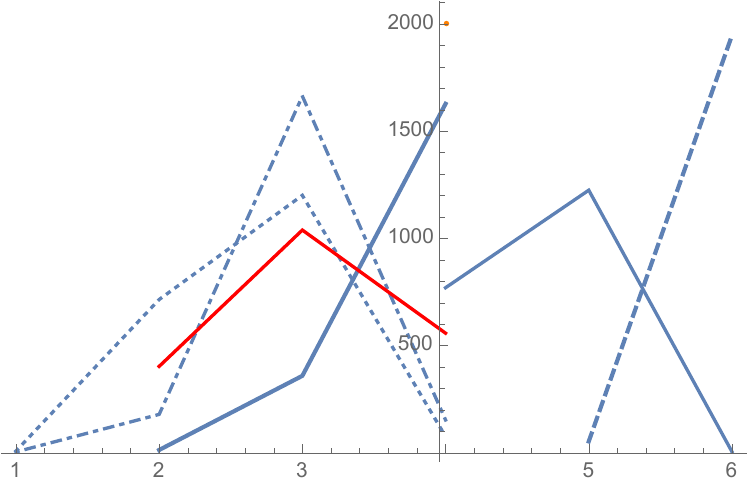}
	\caption{$k=3$}
	\end{subfigure}
  \hfill
	\begin{subfigure}{0.4\linewidth}
	 \includegraphics[width=\textwidth]{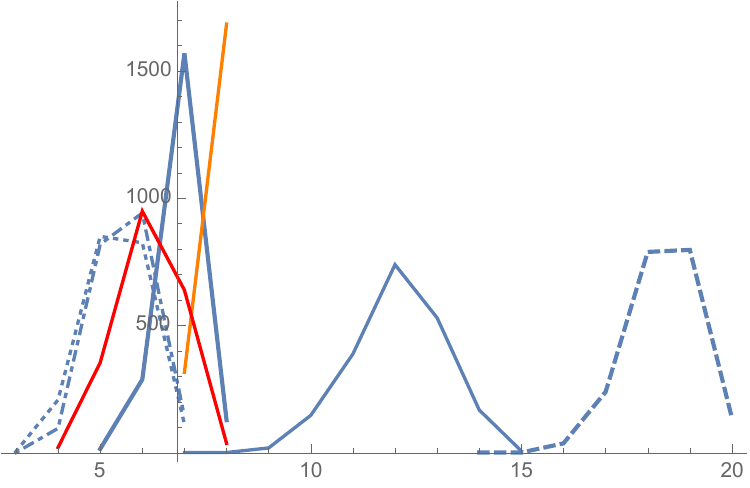}
    \caption{$k=4$}
	\end{subfigure}
  \hfill
	\begin{subfigure}{0.4\linewidth}
	 \includegraphics[width=\textwidth]{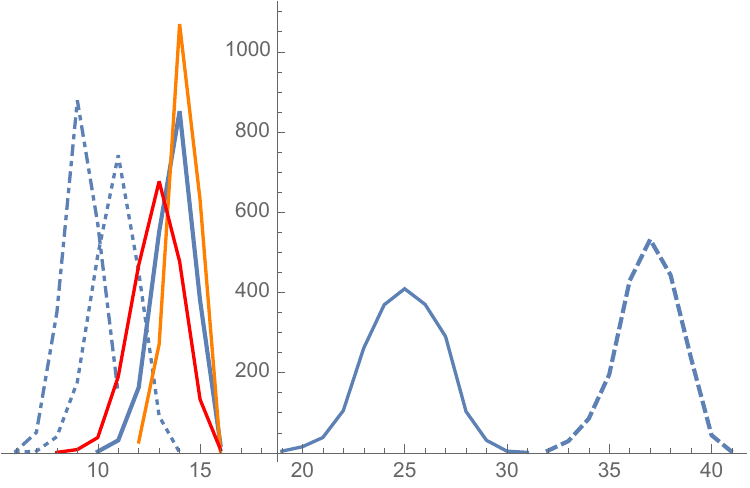}
	\caption{$k=5$}
	\end{subfigure}
  \hfill
	\begin{subfigure}{0.4\linewidth}
	 \includegraphics[width=\textwidth]{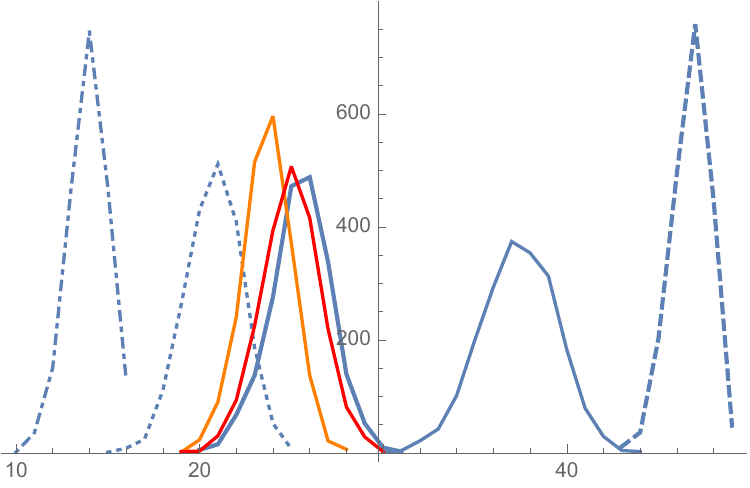}
\caption{$k=6$}
	\end{subfigure}
 \caption{Frequencies for given abundances. Dotted is Black's single-peaked domain, solid is 2004 AGH, dot-dashed is single-crossing, the thick solid line is Fishburn, the orange line is the symmetric/group-separable domain, red is the (5,18)-abundant domain,  and dashed is the unrestricted domain/Impartial culture.}\label{fig1}
\end{figure}

\section{Conclusion}
\label{sec:conclusion}

This paper links the study of local diversity to Condorcet domains. For a given number of alternatives $n$, we determine the Condorcet domains that exhibit the highest possible diversity for their least diverse restriction. The results are valid for an unlimited number $N$ of agents. For $n\geq7$, several of these domains are novel and have not been previously explored, while for smaller values of $n$, Fishburn's alternating domains serve as examples of maximum diversity. Notably, in several instances, the domains with the highest local diversity are not the ones with the greatest number of preference orders, underlining that maximizing the number of preference orders and maximizing diversity are two separate challenges for Condorcet domains.

We have shown that for restrictions to subsets of a fixed size $k$, the abundance cannot be higher than $(k,2^{k-1})$ if $n$ is large enough. This demonstrates that even though Condorcet domains can be much larger than Black's single-peaked domain, their abundance for small subsets cannot be higher than that of Black's single-peaked domain, when the number of alternatives is large. 

\section*{Acknowledgement}
Bei Zhou was funded by the China Scholarship Council (CSC). The Basic Research Program of the National Research University Higher School of Economics partially supported Alexander Karpov. This research utilised Queen Mary's Apocrita HPC facility, supported by QMUL Research-IT. 

\section*{Compliance with Ethical Standards}
\textbf{Conflict of interest} The authors have no relevant financial or non-financial interests that could potentially compromise the integrity, impartiality, or credibility of the research presented in this paper.

\bibliographystyle{econ}
\ifx\undefined\bysame
\newcommand{\bysame}{\hskip.3em \leavevmode\rule[.5ex]{3em}{.3pt}\hskip0.5em}
\fi

\end{document}